\documentclass[12pt]{article}

\pdfoutput=1

\usepackage{amsfonts}
\usepackage{amsmath,amssymb}
\usepackage{graphicx}
\usepackage{amsmath,amssymb}
\usepackage{amssymb}
\usepackage{setspace}

\newtheorem{theorem}{Theorem}[section]
\newtheorem{acknowledgement}[theorem]{Acknowledgement}

\newtheorem{definition}[theorem]{Definition}

\newenvironment{proof}[1][Proof]{\noindent \textbf{#1.} }{\  \rule{0.5em}{0.5em}}
\numberwithin{equation}{section}
\begin{document}

\author{ G. Amendola\thanks{
Department of Mathematics, University of Pisa. E-mail: amendola@dma.unipi.it}  
\ and M. Fabrizio\thanks{
Department of Mathematics, University of Bologna. E-mail: mauro.fabrizio@unibo.it}}
\title{Thermomechanics of damage and fatigue \\
by a phase field model}
\maketitle

\begin{abstract}
In the paper we present an isothermal model for
describing damage and fatigue by the use of the Ginzburg-Landau (G-L)
equation.

Fatigue produces progressive damage, which is related with a variation of
the internal structure of the material. The G-L equation  studies
the evolution of the order parameter, which describes the constitutive
arrangement of the system and, in this framework, the evolution of damage.
The thermodynamic coherence of the model is proved.

In the last part of the work, we extend the results of the paper to a
non-isothermal system, where fatigue contains thermal effects, which
increase the damage of materials.

\end{abstract}

\section{Introduction}

The fatigue notion was founded on the concept of degraded or tired material
and linked to the observation of material damage subjected to cyclic
loading, which never reaches a level sufficient to cause failure in a single
application (see \cite{Cap1}, \cite{Fleck}, \cite{LA}, \cite{Paris} and \cite%
{Schu}).

It is apparent that fatigue produces progressive damage involving plastic
deformation, crack nucleation, creep rupture and finally rapid fracture.
Material damage is the gradual process of mechanical deterioration, that
basically results in a structural component failure (see \cite{Schu} \cite%
{Kork} and \cite{St.Ra.Fu}).

Among the first and most important mathematical models of fatigue, we recall
Caputo's papers \cite{Cap1}, \cite{Cap2} and \cite{Cap3}, where he proposes
a definition of fatigue using the fractional derivative defined by him in 
\cite{Cap4}. This model estimates fatigue as function of the maximum strain
applied, related to the number of cycles and of their frequency.

Fremond and co-authors in \cite{Fre1}, \cite{Fre2}, \cite{Fre3} and \cite%
{Fre4} introduce a scalar variable $\varphi $ describing the damage of the
material and then suggest a definition of fatigue by the mechanical work
that the system has performed starting from its initial or virgin state at
the time $t_{0}$. The variable $\varphi $, in the following called phase
field, describes the internal structural changes related to the damage of
the material. In this paper, we denote by $\mathbf{\tilde{T}}$ the stress
tensor of the virgin material (no damage), then the stress connected with
damage is defined by%
\begin{equation}
\mathbf{T}(x,t)=\left \{ 
\begin{array}{ll}
(1-\varphi )^{2}\mathbf{\tilde{T}}(x,t) & for~\varphi \in \left[ 0,1\right]
\\ 
\mathbf{\tilde{T}}(x,t) & for~\varphi <0 \\ 
\mathbf{0} & for~\varphi >1%
\end{array}%
\right.  \label{1}
\end{equation}%
so that we can denote the equation (\ref{1}) by the following synthetic
formula 
\begin{equation*}
\mathbf{T}(x,t)=(1-\breve{\varphi})^{2}\mathbf{\tilde{T}}(x,t),
\end{equation*}%
where 
\begin{equation*}
\breve{\varphi}=\left \{ 
\begin{array}{ll}
\varphi & for~\varphi \in \left[ 0,1\right] \\ 
1 & for~\varphi >1 \\ 
0 & for~\varphi <0%
\end{array}%
.\right.
\end{equation*}

So, we can define fatigue $\mathcal{F}$ as%
\begin{equation}
\mathcal{F(}x,t)=\int_{t_{0}}^{t}\left[ 1-\breve{\varphi}(x,\tau )\right] 
\mathbf{\tilde{T}}(x,\tau )\cdot \mathbf{L}(x,\tau )dt,  \label{2a}
\end{equation}%
where $\mathbf{v}$ denotes the velocity and $\mathbf{L=}\nabla \mathbf{v}$.

We observe that the stress, expressed by (\ref{1}) and (\ref{2a}), decreases
with the increase of damage $\varphi $. Thus, the stress $\mathbf{T}(x,t)$
drops to zero and we have fracture.

In this framework, the evolution of $\varphi $ will be described inside the
Ginzburg-Landau (G-L) theory, which allows to represent phenomena that
exhibit internal structural changes. Inside of this outline, the G-L theory
provides the differential equation%
\begin{eqnarray}
\rho (x)\frac{\partial \varphi (x,t)}{\partial t} &=&\nabla \cdot \frac{1}{%
\kappa (x)}\nabla \varphi (x,t)-\mathcal{F(}x,t)F^{\prime }(\varphi (x,t)) 
\notag \\
&&-\mathcal{F}_{0}\mathcal{(}x)G^{\prime }(\varphi (x,t))+\rho (x)g(x,t),
\label{3}
\end{eqnarray}%
where $\rho $ denotes the density$,~\kappa $ the diffusivity and $\mathcal{F}%
_{0}$ are suitable positive coefficients, while $F^{\prime }(\varphi )$ and $%
G^{\prime }(\varphi )$ are the derivatives of two suitable potentials $%
F(\varphi )$ and $G(\varphi )$. Finally, $g(x,t)$ denotes the sourge of
damage, that in the following, we shall suppose zero.

From (\ref{3}), we may state that from the properties of the potentials $%
F(\varphi )$ and $G(\varphi )$, when fatigue $\mathcal{F(}x,t)$ overlaps $%
\mathcal{F}_{0}\mathcal{(}x),$ the phase $\varphi $ increases and reaches $1.
$

For this study, we associate to the equation (\ref{3}) the motion equation%
\begin{equation}
\rho (x)\frac{\partial \mathbf{v}(x,t)}{\partial t}=\nabla \cdot \mathbf{T}%
(x,t)+\rho (x)\mathbf{b}(x,t),  \label{4}
\end{equation}%
\newline
where $\rho $ denotes the density and $\mathbf{b}$ the body forces.

In the paper, the Gibbs free energy $W$, reduced to a function of $\varphi $
and $\nabla \varphi $ only, contains the following terms%
\begin{eqnarray}
W(\varphi (x,t),\nabla \varphi (x,t)) &=&\int_{\Omega }\left \{ \frac{1}{%
2\kappa (x)}\left[ \nabla \varphi (x,t)\right] ^{2}\right.   \notag \\
&&\left. +\mathcal{F}(x,t)F(\varphi (x,t))+\mathcal{F}_{0}(x)G(\varphi
(x,t))\right \} dx.  \label{5.1}
\end{eqnarray}%
The equation (\ref{3}) is coherent with the definition of the Gibbs free
energy (\ref{5.1}), because we obtain the second member of (\ref{3}) by $%
W(\varphi ,\nabla \varphi )$. So that, using (\ref{3}) and (\ref{5.1}), we
have%
\begin{equation}
\rho (x)\frac{\partial \varphi (x,t)}{\partial t}=-\frac{\delta }{\delta
\varphi }W(\varphi (x,t),\nabla \varphi (x,t)).  \label{6}
\end{equation}

It is worth to observe that, in the first part of the paper, the notion of
fatigue is defined and studied in a general framework, without to
characterize the constitutive equations of materials. Then, in the next
sections, fatigue and damage are defined for elastic and fading memory
materials. For this fatigue model, it is basic to prove a maximum theorem
for the phase field $\varphi $, that is proved in the Section 6, together to
a uniqueness theorem. Hence, in the Section 7, we study the fatigue behavior
under changes of temperature, because damage can be sensitive to frequent
temperature variations and to high gradients or thermal shocks.

\section{State and process for thermomechanical simple models}

In this section, we regard simple materials non undergoing to fatigue
phenomena. Following the papers Truesdell and Noll \cite{TN}, Noll \cite%
{Noll} and Coleman and Owen \cite{CO} (see also \cite{AFG}), we define the
properties of a thermomechanical solid material by the concept of dynamic
system, which makes use of the notion of state and process.

Let us consider a material point $X$ of a body $\mathcal{B}\footnote{%
For the properties of the body $\mathcal{B}$ see Banfi and Fabrizio \cite%
{Ba-Fa}.}$. At time $t\in \left[ 0,T\right] $ $\left( T\in \mathbb{R}\right) 
$, the \textit{configuration} $C(X,t)$ is given by the pair%
\begin{equation}
C(X,t)=\left( \mathbf{F}(X,t),\theta (X,t)\right) \in Lin^{+}\times \mathbb{R%
}^{+},  \label{7}
\end{equation}%
where $\mathbf{F}$ denotes the \textit{deformation gradient }and $\theta $
the \textit{absolute temperature.} A \textit{thermomechanical process }$P$
of duration $d_{P}\in \mathbb{R}^{+}$ is a piecewise continuous function in
the interval $\left[ 0,d_{P}\right) $, given by the triplet%
\begin{equation}
P(t)=\left( \mathbf{L}(t),\dot{\theta}(t),\mathbf{g}(t)\right) \in \Xi ,~\
t\in \left[ 0,d_{P}\right)  \label{8}
\end{equation}%
with $\mathbf{L}(t)\in Lin$ the velocity gradient, $\dot{\theta}(t)=\frac{%
d\theta (t)}{dt}$ and $\mathbf{g}(t)=\nabla \theta (t)$, while $\Xi
=Lin\times \mathbb{R\times R}^{3}.$

In the following, we denote by $P_{t}$ the process restricted to the
interval $\left[ 0,t\right) $, with $t\in \left[ 0,d_{P}\right) .$

A thermomechanical material in a point $X\in \mathcal{B}\ ~$can be view as a
dynamic system, where the process $P$ is the input, while the output $\hat{U}
$ is defined by the functions%
\begin{equation*}
\hat{U}(t)=\left( \mathbf{\tilde{T}}(t),h(t),\mathbf{q}(t)\right) ,~\ t\in %
\left[ 0,d_{P}\right)
\end{equation*}%
with$\  \mathbf{\tilde{T}}$ the stress tensor,$~h$ the internal heat power
and $\mathbf{q}$ the heat flux.

As defined in \cite{Noll}, \cite{CO} and \cite{AFG}, a simple
thermomechanical dynamic system is defined by the set $(\mathcal{C},\mathcal{%
\Pi },\Sigma ,\hat{\rho})$, where

a - $\mathcal{C}$ is an open, connected subset of $Lin^{+}\times \mathbb{R}%
^{+}$ of configurations $C(X)$,

b - $\Pi $ is the set of smooth thermokinetic processes,

c - $\Sigma $ is a metric space, defined by the state set,

d - $\hat{\rho}:\Sigma \times \Pi \rightarrow \Sigma $ is the evolution
function, such that%
\begin{equation}
\hat{\rho}(\sigma ^{i},P)=\sigma ^{f},  \label{9}
\end{equation}%
where $\sigma ^{i}\in \Sigma $ is the initial state and $\sigma ^{f}\in
\Sigma $ is the state obtained from $\sigma ^{i}$ by the process $P.$
Moreover, the map $\hat{\rho}$ is such that, if $\sigma \in \Sigma $ and $%
P^{1},P^{2}\in \Pi $, then 
\begin{equation}
\hat{\rho}(\sigma ,P^{1}\ast P^{2})=\hat{\rho}(\hat{\rho}(\sigma
,P^{1}),P^{2})\text{,}  \label{10}
\end{equation}%
where $P^{1}\ast P^{2}$ denotes the composition of $P^{1}$ with $P^{2}.$

So that, a thermomechanical material is defined by the functional $\hat{U}%
:\Sigma \times \Xi \rightarrow Sym\times \mathbb{R\times R}^{3}$, whereby%
\begin{eqnarray}
\mathbf{\tilde{T}}(t) &=&\mathbf{\tilde{T}}(\sigma (t),P(t)),  \notag \\
h(t) &=&\hat{h}(\sigma (t),P(t)),  \label{11} \\
\mathbf{q}(t) &=&\mathbf{\hat{q}}(\sigma (t),P(t)).  \notag
\end{eqnarray}

Moreover, we define the null process of duration $d_{P^{0}}$ 
\begin{equation*}
P^{0}(t)=(\mathbf{0}\text{,}0,\mathbf{0}),\  \text{for any }~t\in \left[
0,d_{P^{0}}\right) .
\end{equation*}%
Finally, the null state $\sigma ^{0\text{ }}$is such that%
\begin{equation*}
\hat{\rho}(\sigma ^{0},P^{0})=\sigma ^{0}
\end{equation*}%
for any null process $P^{0}$ of duration $d_{p^{0}}\in \mathbb{R}^{+}.$

\section{Damage and fatigue}

Fatigue describes the material wear resulting from cyclic loading. While
damage features the progressive mechanical deterioration that can involve
crack nucleation, dislocation, plastic deformation, such that ultimately can
bring in component failure.

Because, damage and fatigue are related to one another, we need to
supplement the notion of state with a new variable $\varphi \in \mathbb{R}$,
describing the changes of the internal material structure consequent to
damage and related to fatigue. Hence, in this framework the variable $%
\varphi $ represents variations of internal material structure. Therefore,
the notion of state and process will be modified by the new definitions%
\begin{equation*}
\tilde{\sigma}=(\sigma ,\varphi ,\nabla \varphi ),~\  \tilde{P}(t)=(P(t),\dot{%
\varphi}(t),\nabla \dot{\varphi}),~~t\in \left[ 0,d_{P}\right) ,
\end{equation*}%
while the response of the material $U(t)$ assumes the new representation%
\begin{equation}
\hat{U}(t)=\left( \mathbf{T}(t),h(t),\mathbf{q}(t),\mathcal{F}(t)\right) \
,~\ t\in \left[ 0,d_{P}\right) ,  \label{12}
\end{equation}%
where $\mathbf{T}$ is defined as in (\ref{1}), while the fatigue $\mathcal{F}%
(t)$ is given in (\ref{2a}), which in this framework we write%
\begin{equation}
\mathcal{F}(\tilde{\sigma}^{i},\tilde{P}_{t})=\int_{t_{0}}^{t}\left[ 1-%
\breve{\varphi}(\tau )\right] \mathbf{\tilde{T}}(\sigma (\tau ),P(\tau
))\cdot \mathbf{L}(\tau )d\tau ,  \label{13}
\end{equation}%
\bigskip where $\tilde{\sigma}^{i}$ is the initial state at time $t_{0}$,
while $\tilde{\sigma}^{i}=(\sigma ^{i},0,0)$ and $\sigma (\tau )=\hat{\rho}%
(\sigma ^{i},P_{\tau }).$

In the following of this paper, we shall restrict the study to mechanical
processes only. Whereby, the process is given by $\tilde{P}(t)=(\mathbf{L}%
(t),\dot{\varphi}(t))$ and the response by $U(t)=(\mathbf{T}(t),\mathcal{F}%
(t)).$

For the whole study of the problem, we need to introduce an equation for the
phase field $\varphi .$ Because this variable describes the variation of the
internal structure of the material, the natural representation is given by
the Ginzburg-Landau equation (\ref{3}) with $g=0$, which for this problem
assumes the form%
\begin{equation}
\rho (x)\frac{\partial \varphi (x,t)}{\partial t}=\nabla \cdot \frac{1}{%
\kappa (x)}\nabla \varphi (x,t)-\mathcal{F}(\tilde{\sigma}^{0},\tilde{P}%
_{t})F^{\prime }(\varphi (x,t))-\mathcal{F}_{0}\mathcal{(}x)G^{\prime
}(\varphi (x,t)),  \label{14}
\end{equation}%
where $\kappa (x)$ denotes the diffusivity. In this paper, we choose the two
potentials $F$ and $G$ given by

\begin{equation}
F(\varphi )=-\breve{\varphi},~\ G(\varphi )=\left \{ 
\begin{array}{ll}
\varphi ^{2}-\frac{\varphi ^{3}}{6} & 0\leq \varphi \leq 1 \\ 
\frac{5}{6} & \varphi >1 \\ 
0 & \varphi <0%
\end{array}%
.\right.  \label{15a}
\end{equation}%
In the following, we shall denote

\begin{equation*}
F^{\prime }(\varphi )=\left \{ 
\begin{array}{ll}
-1 & 0\leq \varphi \leq 1 \\ 
0 & \varphi >1 \\ 
0 & \varphi <0%
\end{array}%
\right. :=\left \langle -1\right \rangle \leq 0
\end{equation*}%
\begin{equation*}
G^{\prime }(\varphi )=\left \{ 
\begin{array}{ll}
2\varphi -\frac{\varphi ^{2}}{2} & 0\leq \  \varphi \leq 1 \\ 
0 & \varphi >1 \\ 
0 & \varphi <0%
\end{array}%
\right. :=\left \langle 2\varphi -\frac{\varphi ^{2}}{2}\right \rangle \geq
0.
\end{equation*}

Together with the equation (\ref{14}), we consider the motion equation, that
is%
\begin{equation}
\rho (x)\frac{\partial \mathbf{v}(x,t)}{\partial t}=\nabla \cdot \mathbf{T}%
(x,t)+\rho (x)\mathbf{b}(x,t).  \label{16}
\end{equation}

Of course, the material properties will be defined by the constitutive
equation on the stress tensor $\mathbf{T}(x,t)$, which in this framework is
fixed by the equation (\ref{1}), that we may even write in the equivalent
form 
\begin{equation}
\mathbf{T}(\tilde{\sigma}(x,t),\tilde{P}(x,t))=\left[ 1-\breve{\varphi}(x,t)%
\right] ^{2}\mathbf{\tilde{T}}(\sigma (x,t),P(x,t))  \label{17}
\end{equation}%
so%
\begin{equation}
\mathbf{T}(\tilde{\sigma}(x,t),\tilde{P}(x,t))\cdot \nabla \mathbf{v}(x,t)=%
\left[ 1-\breve{\varphi}(x,t)\right] ^{2}\mathbf{\tilde{T}}(\sigma
(x,t),P(x,t))\cdot \nabla \mathbf{v}(x,t);  \label{17b}
\end{equation}%
then, if we suppose to consider processes leaving from an initial state $%
\tilde{\sigma}^{i}=(\sigma ^{i},0,0)$ at the time $t_{0}$, the equation (\ref%
{17b}) can be written%
\begin{equation}
\mathbf{T}(\tilde{\sigma}(t),\tilde{P}(t))\cdot \nabla \mathbf{v}(t)=\left[
1-\breve{\varphi}(t)\right] \frac{d}{dt}\int_{t_{0}}^{t}\left[ 1-\breve{%
\varphi}(\tau )\right] \mathbf{\tilde{T}}(\sigma (\tau ),P(\tau ))\cdot
\nabla \mathbf{v}(\tau )d\tau .  \label{18a}
\end{equation}

Hence, by the definition (\ref{13}) of fatigue we have%
\begin{equation}
\mathbf{T}(\tilde{\sigma}(t),\tilde{P}(t))\cdot \nabla \mathbf{v}(t)=\left[
1-\breve{\varphi}(t)\right] \mathbf{\tilde{T}}(\sigma (t),P(t))\cdot \nabla 
\mathbf{v}(t)-\breve{\varphi}(t)\frac{d}{dt}\mathcal{F}(\tilde{\sigma}^{i},%
\tilde{P}_{t}).  \label{18b}
\end{equation}

\section{Dissipation Principle}

In order to claim the dissipation principle for any $x\in \mathcal{B}$, we
define the internal mechanical power $\mathcal{P}_{m}^{i}(t)=\mathbf{T}(%
\tilde{\sigma}(t),\tilde{P}(t))\cdot \nabla \mathbf{v}(t)$ and the internal
structural power $\mathcal{P}_{s}^{i}(t)~$at time $t$ 
\begin{equation}
\mathcal{P}_{m}^{i}(\tilde{\sigma}(t),\tilde{P}(t)=\left[ 1-\breve{\varphi}%
(t)\right] \mathbf{\tilde{T}}(\sigma (t),P(t))\cdot \nabla \mathbf{v}(t)-%
\breve{\varphi}(t)\frac{d}{dt}\mathcal{F}(\tilde{\sigma}^{i},\tilde{P}_{t})%
\text{,}  \label{19}
\end{equation}%
\begin{equation}
\mathcal{P}_{s}^{i}(\tilde{\sigma}(t),\tilde{P}(t)=\rho \left[ \dot{\varphi}%
(t)\right] ^{2}+\left \{ \frac{1}{2\kappa }\left[ \nabla \varphi (t)\right]
^{2}\right \} ^{\cdot }+\mathcal{F}_{0}\dot{G}(\varphi (t))+\mathcal{F}(%
\tilde{\sigma}^{i},\tilde{P}_{t})\dot{F}(\varphi (t))\text{,}  \label{20}
\end{equation}%
where the dot $^{\cdot }$ denotes the partial time derivative.

\textbf{Dissipation Principle. }There exists a state function $\psi (\tilde{%
\sigma})$, called free energy, such that%
\begin{equation}
\rho \dot{\psi}(\tilde{\sigma}(t))\leq \mathcal{P}_{m}^{i}(t)+\mathcal{P}%
_{s}^{i}(t)\text{ .}  \label{21}
\end{equation}

By using (\ref{19}) and (\ref{20}), we obtain from (\ref{21}) 
\begin{eqnarray}
&&\rho \dot{\psi}(\tilde{\sigma}(t))\leq \left[ 1-\breve{\varphi}(t)\right] 
\mathbf{\tilde{T}}(\sigma (t),P(t))\cdot \nabla \mathbf{v}(t)-\breve{\varphi}%
(t)\frac{d}{dt}\mathcal{F}(\tilde{\sigma}^{i},\tilde{P}_{t})  \notag \\
&&\quad +\rho \left[ \dot{\varphi}(t)\right] ^{2}+\left \{ \frac{1}{2\kappa }%
\left[ \nabla \varphi (t)\right] ^{2}\right \} ^{\cdot }+\mathcal{F}_{0}\dot{G%
}(\varphi (t))+\mathcal{F}(\tilde{\sigma}^{i},\tilde{P}_{t})\dot{F}(\varphi
(t))\quad   \label{22a}
\end{eqnarray}%
and hence, on account of (\ref{15a})$_{1}$, we have 
\begin{eqnarray}
&&\rho \dot{\psi}(\tilde{\sigma}(t))\leq \left[ 1-\breve{\varphi}(t)\right] 
\mathbf{\tilde{T}}(\sigma (t),P(t))\cdot \nabla \mathbf{v}(t)+\rho \left[ 
\dot{\varphi}(t)\right] ^{2}  \notag \\
&&\quad +\left \{ \frac{1}{2\kappa }\left[ \nabla \varphi (t)\right]
^{2}\right \} ^{\cdot }+\mathcal{F}_{0}\dot{G}(\varphi (t))+\frac{d}{dt}\left[
\mathcal{F}(\tilde{\sigma}^{i},\tilde{P}_{t})F(\varphi (t)\right] ,\quad 
\label{23}
\end{eqnarray}%
or%
\begin{eqnarray}
&&\rho \dot{\psi}(\tilde{\sigma}(t))\leq \left[ 1-\breve{\varphi}(t)\right] 
\mathbf{\tilde{T}}(\sigma (t),P(t))\cdot \nabla \mathbf{v}(t)+\rho \left[ 
\dot{\varphi}(t)\right] ^{2}  \notag \\
&&\quad +\frac{d}{dt}\left \{ \frac{1}{2\kappa }\left[ \nabla \varphi (t)%
\right] ^{2}+\mathcal{F}_{0}G(\varphi (t))+\mathcal{F}(\tilde{\sigma}^{i},%
\tilde{P}_{t})F(\varphi (t))\right \} \text{.\quad }  \label{24}
\end{eqnarray}

This inequality may be written in the equivalent form%
\begin{eqnarray}
&&\rho \dot{\psi}(\tilde{\sigma}(t))\leq \rho \left[ \dot{\varphi}(t)\right]
^{2}+  \notag \\
&&\quad \frac{d}{dt}\left \{ \frac{1}{2\kappa }\left[ \nabla \varphi (t)%
\right] ^{2}+\mathcal{F}_{0}G(\varphi (t))+\mathcal{F}(\tilde{\sigma}^{i},%
\tilde{P}_{t})F(\varphi (t))\right \} \text{.\quad }  \label{24b}
\end{eqnarray}

Now, we denote by $\psi _{V}$ the free energy related to virgin material (no
damage), so that%
\begin{equation}
\rho \dot{\psi}_{V}(\sigma (x,t))\leq \mathbf{\tilde{T}}(\sigma
(x,t),P(x,t))\cdot \nabla \mathbf{v}(x,t).  \label{25}
\end{equation}

Then, from (\ref{24}) and (\ref{25}) we have that the free energy $\psi $
can be written as%
\begin{equation*}
\psi (\tilde{\sigma}(t))=\psi _{D}(\tilde{\sigma}(t)+\psi _{V}(\sigma (t),
\end{equation*}%
where $\psi _{D}$ is the damage component of free energy.

Otherwise, if we define the pseudo fatigue energy $\psi _{F}$ as 
\begin{equation}
\psi _{F}(\tilde{\sigma}(t))=\frac{1}{\rho }\left \{ \frac{1}{2\kappa }\left[
\nabla \varphi (t)\right] ^{2}+\mathcal{F}_{0}G(\varphi (t))+\mathcal{F}(%
\tilde{\sigma}^{i},\tilde{P}_{t})F(\varphi (t))\right \} ;  \label{27}
\end{equation}%
then, from (\ref{24}), we obtain 
\begin{equation}
\rho \left[ \dot{\psi}(\tilde{\sigma}(t))-\dot{\psi}_{F}(\tilde{\sigma}(t))%
\right] \leq \left[ 1-\breve{\varphi}(t)\right] \mathbf{\tilde{T}}(\sigma
(t),P(t))\cdot \nabla \mathbf{v}(t)+\rho \left[ \dot{\varphi}(t)\right] ^{2}.
\label{27-1}
\end{equation}

Of course, once assigned the constitutive equations of virgin material, the
mathematical shape of $\psi (\tilde{\sigma}(t)$ can be well defined.

Otherwise, we may study the behavior of the pseudo energy $\psi _{F}$ with $%
F $ and $G$ defined in (\ref{15a}). So that, if we consider the graph of $%
\psi _{F}$, as a function only of the field $\varphi $, varying the fatigue,
we obtain different shape of $\psi _{F}$ represented in Fig.1 and Fig.2.

\begin{center}
\includegraphics[scale=0.22]{{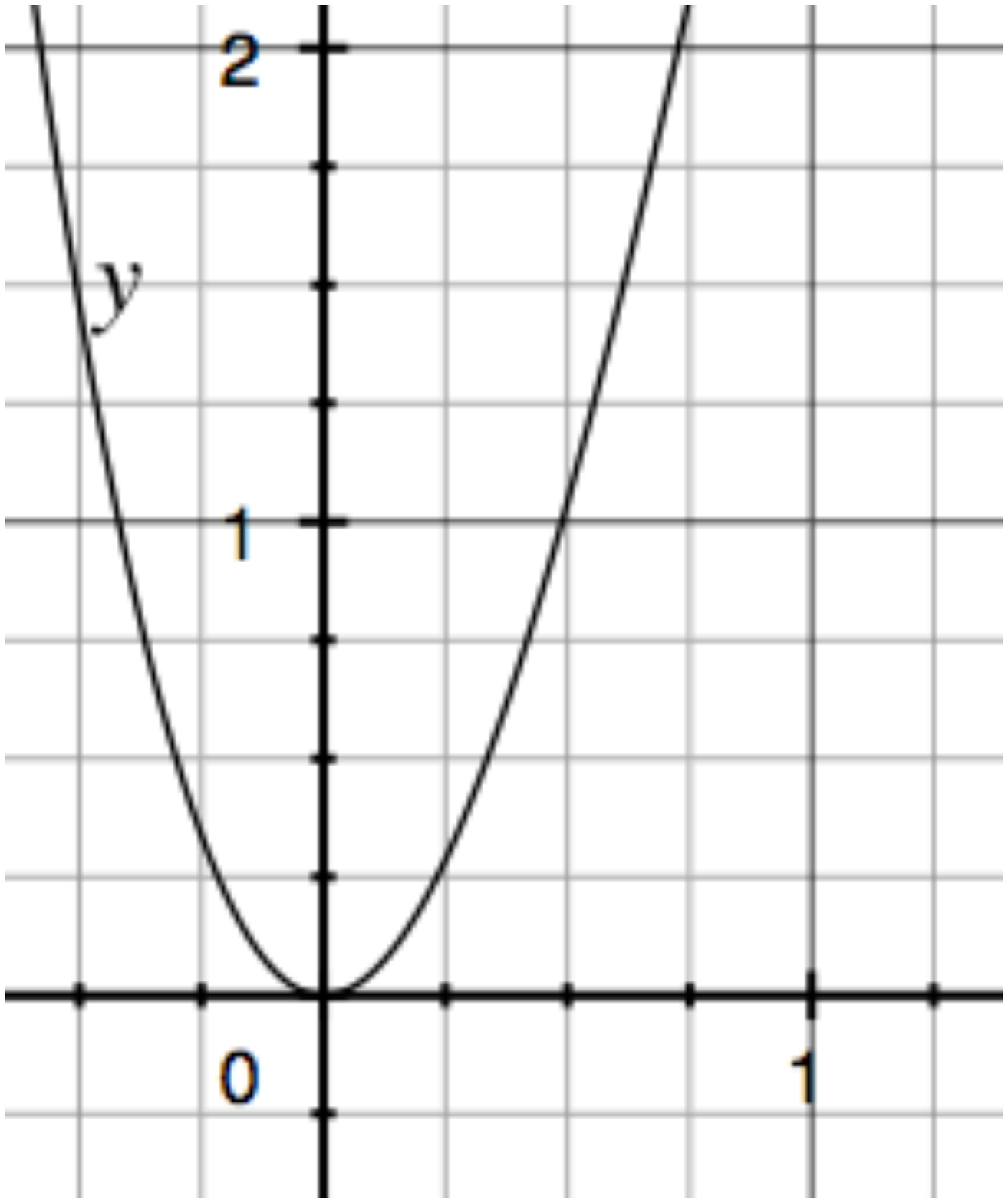}} %
\includegraphics[scale=0.22]{{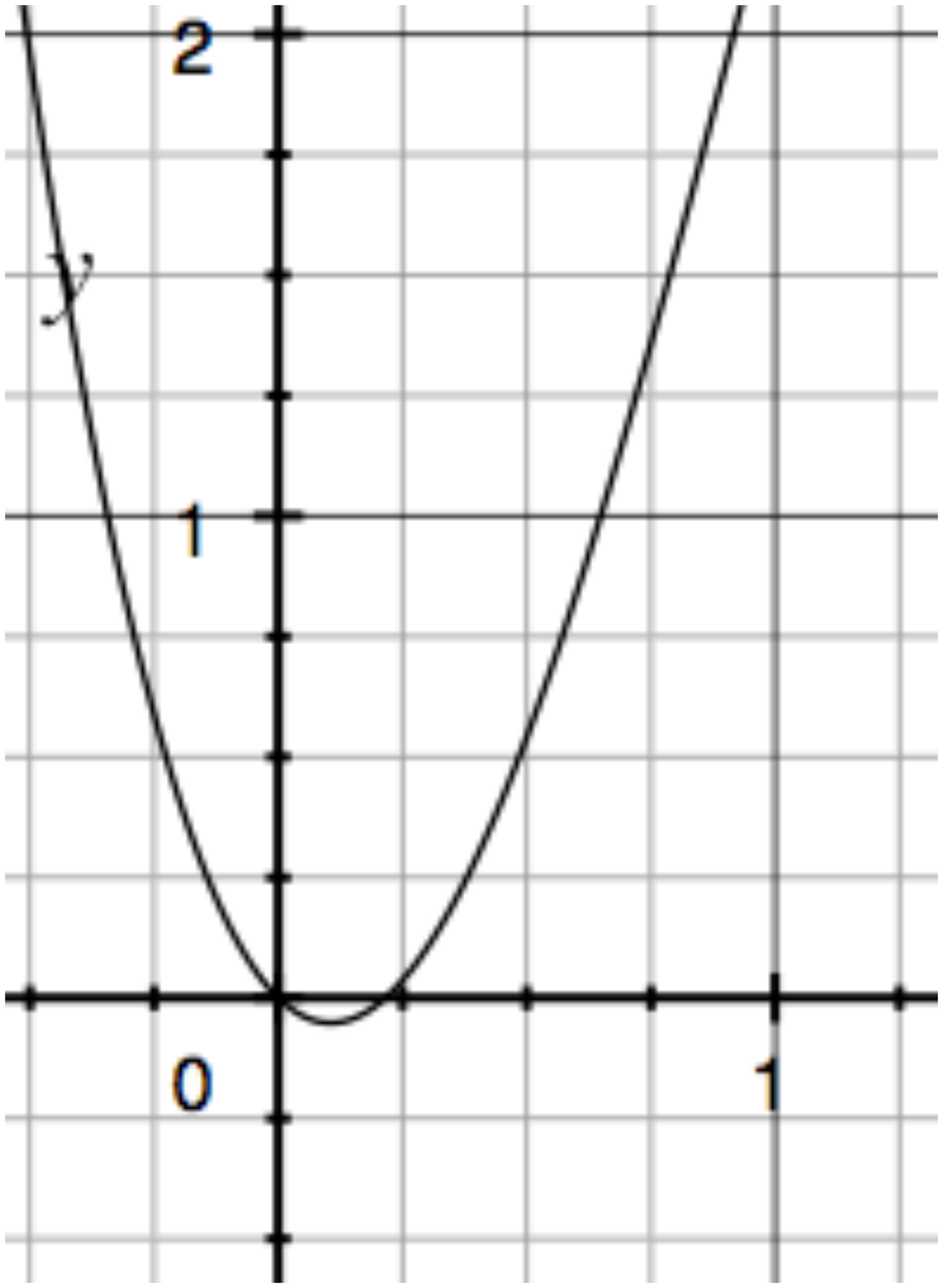}} %
\includegraphics[scale=0.225]{{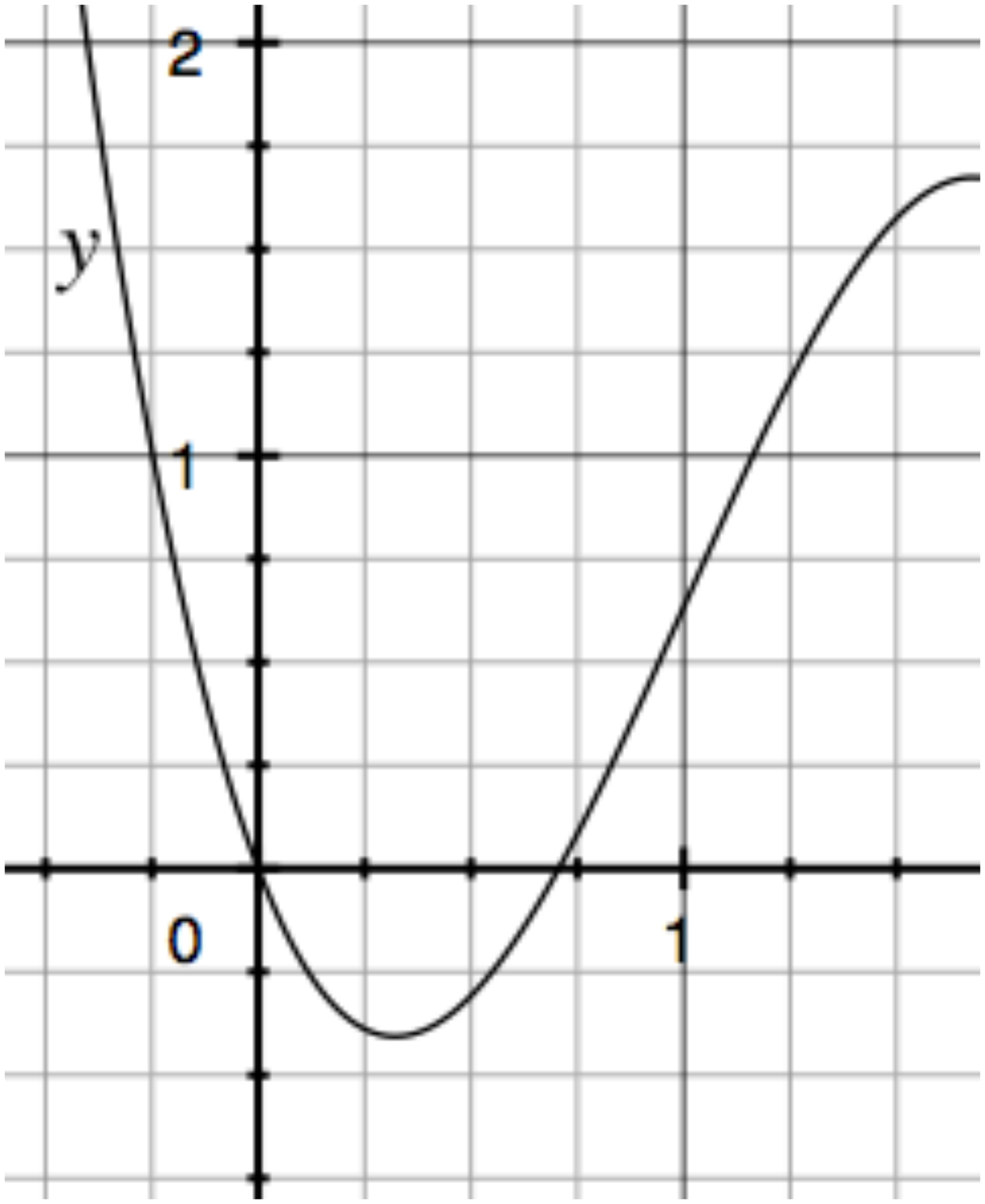}}

Fig.1

The free energies corresponding to different fatigues, with minima in $%
\varphi _{1}=0<\varphi _{2}<\varphi _{3}$

\includegraphics[scale=0.23]{{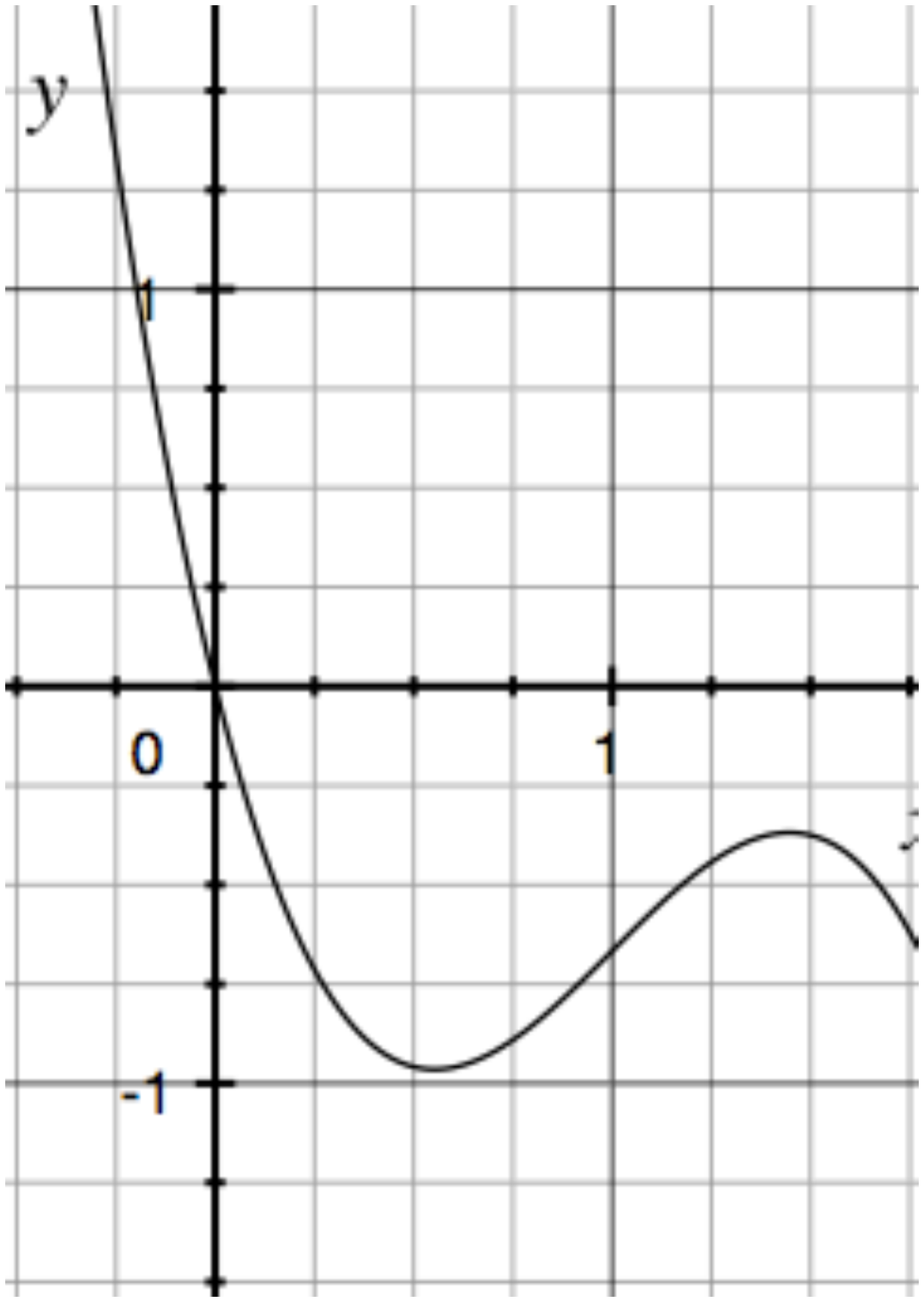}} %
\includegraphics[scale=0.225]{{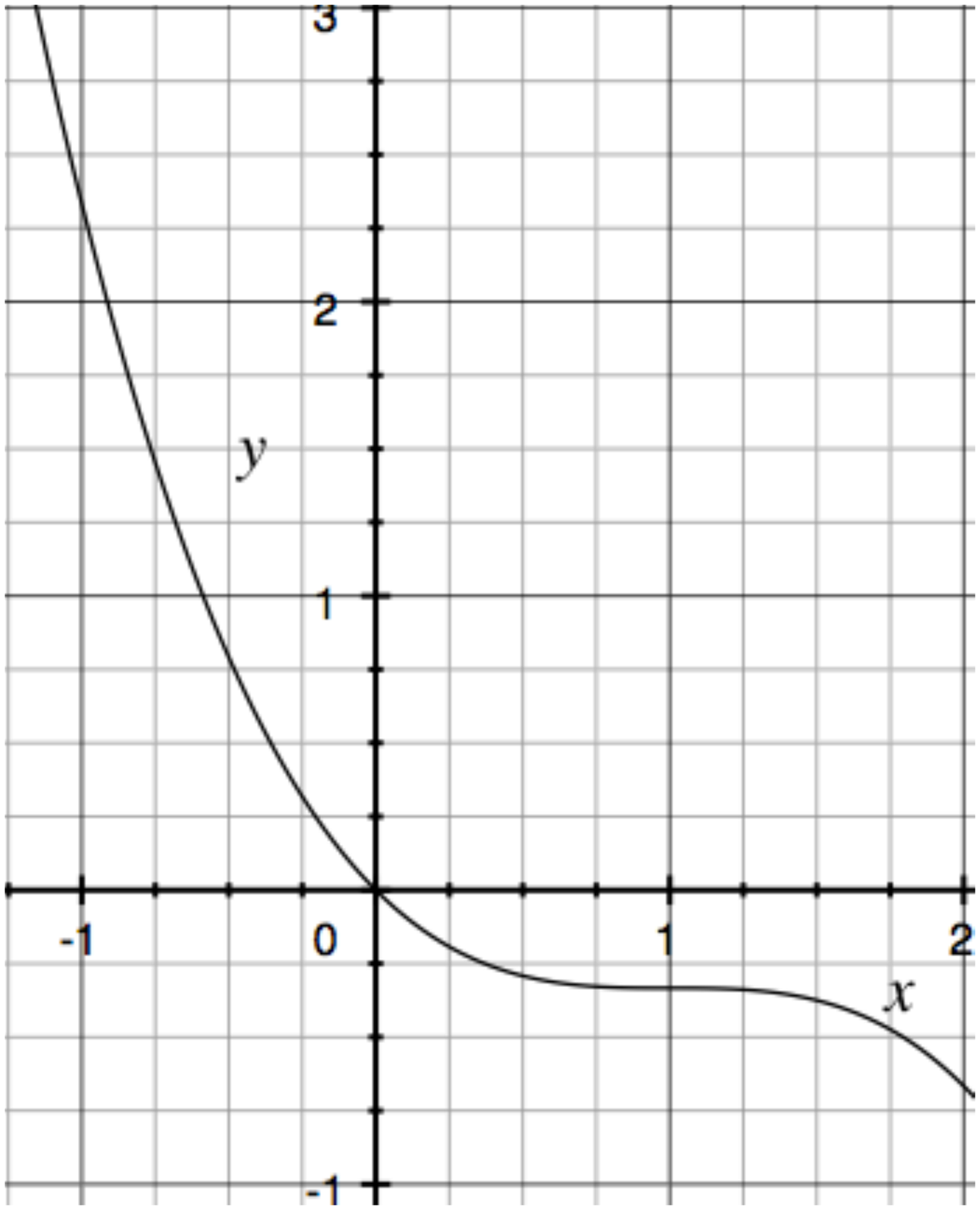}}

Fig.2

In this pictures, the free energies with minima in $\varphi _{4}<\varphi
_{5}=1$
\end{center}

Therefore, the increase of the fatigue shifts the minimum of the free energy 
$\psi _{F}$ toward the extreme value $\varphi =1.$ This involves the
material breaking due to energy instability.

\section{Fatigue for elastic virgin material}

In this section, we consider fatigue model for a material whose virgin state
is elastic. In such a case, the stress tensor $\mathbf{\tilde{T}}$ assumes
the following form%
\begin{equation}
\mathbf{\tilde{T}}(\sigma (t),P(t))=\mathbf{A}\nabla \mathbf{u}(t),
\label{27a}
\end{equation}%
where $\mathbf{A}$ is a positive and symmetric fourth order tensor and $%
\mathbf{u}$ denotes the displacement vector. So fatigue is given by%
\begin{eqnarray}
&&\mathcal{F}(\varphi (t),\nabla \mathbf{u}(t))=\int_{t_{0}}^{t}\left[ 1-%
\breve{\varphi}(\tau )\right] \mathbf{A}\nabla \mathbf{u}(\tau )\cdot \nabla 
\mathbf{\dot{u}}(\tau )d\tau   \notag \\
&&\quad =\frac{1}{2}\left[ 1-\breve{\varphi}(t)\right] \mathbf{A}\nabla 
\mathbf{u}(t)\cdot \nabla \mathbf{u}(t)+\frac{1}{2}\int_{t_{0}}^{t}\left[ 
\breve{\varphi}(\tau )\right] ^{\cdot }\mathbf{A}\nabla \mathbf{u}(\tau
)\cdot \nabla \mathbf{u}(\tau )d\tau ,\qquad   \label{27c}
\end{eqnarray}%
where we have supposed that $\varphi (t_{0})=0$ and $\nabla \mathbf{u}%
(t_{0})=\mathbf{0.}$

Then, from (\ref{27}), (\ref{15a})$_{1}$ and (\ref{27c})$_{2}$, we obtain
the free energy $\psi _{F}(\tilde{\sigma}(t))$ as a function of $(\nabla 
\mathbf{u},\varphi ,\nabla \varphi )$%
\begin{eqnarray*}
&&\psi _{F}(\nabla \mathbf{u}(t),\varphi (t),\nabla \varphi (t))=\frac{1}{%
\rho }\left \{ -\breve{\varphi}(t)\mathcal{F}(\tilde{\sigma}^{i},\tilde{P}%
_{t})+\frac{1}{2\kappa }\left[ \nabla \varphi (t)\right] ^{2}\right.  \\
&&\quad \left. +\mathcal{F}_{0}G(\varphi (t))\right \} =-\frac{1}{2\rho }%
\breve{\varphi}(t)\left \{ \left[ 1-\breve{\varphi}(t)\right] \mathbf{A}%
\nabla \mathbf{u}(t)\cdot \nabla \mathbf{u}(t)+\right.  \\
&&\quad \left. \int_{t_{0}}^{t}\left[ \breve{\varphi}(\tau )\right] ^{\cdot }%
\mathbf{A}\nabla \mathbf{u}(\tau )\cdot \nabla \mathbf{u}(\tau )d\tau
\right \} +\frac{1}{\rho }\left \{ \frac{1}{2\kappa }\left[ \nabla \varphi (t)%
\right] ^{2}+\mathcal{F}_{0}G(\varphi (t))\right \} .
\end{eqnarray*}

While the inequality (\ref{24}), using (\ref{27c})$_{1}$, assumes the form%
\begin{eqnarray*}
&&\rho \dot{\psi}(\tilde{\sigma}(t))\leq \left[ 1-\breve{\varphi}(t)\right] 
\mathbf{A}\nabla \mathbf{u}(t)\cdot \nabla \mathbf{\dot{u}}(t)+\rho \left[ 
\dot{\varphi}(t)\right] ^{2} \\
&&\quad +\frac{d}{dt}\left \{ -\breve{\varphi}(t)\int_{t_{0}}^{t}\left[ 1-%
\breve{\varphi}(\tau )\right] \mathbf{A}\nabla \mathbf{u}(\tau )\cdot \nabla 
\mathbf{\dot{u}}(\tau )d\tau +\frac{1}{2\kappa }\left[ \nabla \varphi (t)%
\right] ^{2}\right.  \\
&&\quad \, \left. +\mathcal{F}_{0}G(\varphi (t))\right \} =-\breve{\varphi}(t)%
\mathbf{A}\nabla \mathbf{u}(t)\cdot \nabla \mathbf{\dot{u}}(t)+\rho \left[ 
\dot{\varphi}(t)\right] ^{2}+\frac{1}{2}\frac{d}{dt}\left \{ [1-\breve{\varphi%
}(t)\right. \qquad  \\
&&\quad +\breve{\varphi}^{2}(t)]\mathbf{A}\nabla \mathbf{u}(t)\cdot \nabla 
\mathbf{u}(t)-\breve{\varphi}(t)\int_{t_{0}}^{t}\left[ \breve{\varphi}(\tau )%
\right] ^{\cdot }\mathbf{A}\nabla \mathbf{u}(\tau )\cdot \nabla \mathbf{u}%
(\tau )d\tau  \\
&&\quad \left. +\frac{1}{\kappa }\left[ \nabla \varphi (t)\right] ^{2}+2%
\mathcal{F}_{0}G(\varphi (t))\right \} .
\end{eqnarray*}%
\newline
Hence, the free energy for these materials is given by%
\begin{eqnarray*}
&&\psi (\nabla \mathbf{u}(t),\varphi (t),\nabla \varphi (t))=\frac{1}{2\rho }%
\left \{ \left[ 1-\breve{\varphi}(t)+\breve{\varphi}^{2}(t)\right] \mathbf{A}%
\nabla \mathbf{u}(t)\cdot \nabla \mathbf{u}(t)\right.  \\
&&\quad \left. -\breve{\varphi}(t)\int_{t_{0}}^{t}\dot{\varphi}(\tau )%
\mathbf{A}\nabla \mathbf{u}(\tau )\cdot \nabla \mathbf{u}(\tau )d\tau +\frac{%
1}{\kappa }\left[ \nabla \varphi (t)\right] ^{2}+2\mathcal{F}_{0}G(\varphi
(t)))\right \} .\qquad 
\end{eqnarray*}

\section{Differential system and maximum theorem}

Now, we consider the differential system for a fatigue mechanical problem
given by the equations (\ref{14}) and (\ref{16}), inside a smooth domain $%
\Omega \subset \mathbb{R}^{3}$ and in a time domain $I=\left[ 0,T\right]
\subset \mathbb{R}.$

In the case of an elastic material $\mathbf{\tilde{T}}$ and $\mathcal{F}(%
\tilde{\sigma}^{0},\tilde{P}_{t})$ are defined in (\ref{27a}) and (\ref{27c}%
), so, by using (\ref{17}, we have the differential system

\begin{eqnarray}
&&\rho (x)\frac{\partial \varphi (x,t)}{\partial t}=\nabla \cdot \frac{1}{%
\kappa (x)}\nabla \varphi (x,t)-\mathcal{F}_{0}\mathcal{(}x)\left \langle
2\varphi (x,t)-\frac{\varphi ^{2}(x,t)}{2}\right \rangle  \notag \\
&&\qquad \qquad -\left \langle -1\right \rangle \int_{0}^{t}\left[ 1-\breve{%
\varphi}(x,\tau )\right] \mathbf{A}(x)\nabla \mathbf{u}(x,\tau )\cdot \nabla 
\mathbf{\dot{u}}(x,\tau )d\tau ,  \label{33b} \\
&&\rho (x)\frac{\partial ^{2}\mathbf{u}(x,t)}{\partial t^{2}}=\nabla \cdot
\left \{ \left[ 1-\breve{\varphi}(t)\right] ^{2}\mathbf{A}(x)\nabla \mathbf{u%
}(x,t)\right \} +\rho (x)\mathbf{b}(x,t),  \label{33c}
\end{eqnarray}

with the initial and boundary conditions 
\begin{equation}
\mathbf{u}(x,0)=\mathbf{u}_{0}(x),~\mathbf{v}(x,0)=\mathbf{v}%
_{0}(x),~\varphi (x,0)=\varphi _{0}(x),  \label{33d}
\end{equation}%
\begin{equation}
\left. \mathbf{u}(x,t)\right \vert _{\partial \Omega }=\mathbf{0},~\  \left.
\nabla \varphi (x,t)\right \vert _{\partial \Omega }=0.  \label{33f}
\end{equation}

In order to define the week solutions of the problem (\ref{33b})-\ref{33f}),
we introduce the set of the functions $\mathbf{u}\in \mathcal{H}_{\mathbf{u}%
}:=H^{3/2}(0,T;L^{2}(\Omega )\cap L^{2}(0,T;H_{0}^{1}(\Omega ))$ and $%
\varphi \in \mathcal{H}_{\varphi }:=H^{1}(0,T;L^{2}(\Omega ))\cap
L^{2}(0,T;H_{\nabla }^{1}(\Omega )),$ where%
\begin{equation*}
H_{0}^{1}(\Omega )=\left \{ \mathbf{u}(x)\in H^{1}(\Omega );\left. \mathbf{u}%
(x)\right \vert _{\partial \Omega }=0\right \}
\end{equation*}%
and%
\begin{equation*}
H_{\nabla }^{1}(\Omega )=\left \{ \varphi (x)\in H^{1}(\Omega );\left.
\nabla \varphi (x)\cdot \mathbf{n}\right \vert _{\partial \Omega }=0\right
\} .
\end{equation*}%
Moreover, we suppose $\rho (x)$, $\kappa (x)$, $\mathbf{A}(x)\in L^{\infty
}(\Omega )$.

\begin{definition}
A pair $(\varphi ,\mathbf{u)\in ~}\mathcal{H}_{\varphi }\times \mathcal{H}_{%
\mathbf{u}}$ is called week solution of the problem (\ref{33b})-\ref{33c}),
with initial conditions (\ref{33d}) and boundary conditions (\ref{33f}), if
satisfies, for all $\psi \in \mathcal{H}_{\varphi }$ and $\mathbf{w}\in 
\mathcal{H}_{\mathbf{u}}$, the integral forms 
\begin{eqnarray}
&&\int_{0}^{T}\int_{\Omega }\left \{ \rho (x)\right. \dot{\varphi}(x,t)\dot{%
\psi}(x,t)+\frac{1}{\kappa (x)}\nabla \varphi (x,t)\cdot \nabla \dot{\psi}%
(x,t)(x,t)  \notag \\
&&\quad +\left \langle -1\right \rangle \int_{0}^{t}\left[ 1-\breve{\varphi}%
(x,\tau )\right] \mathbf{A}(x)\nabla \mathbf{u}(x,\tau )\cdot \nabla \mathbf{%
\dot{u}}(x,\tau )d\tau \dot{\psi}(x,t)  \notag \\
&&\quad \left. +\mathcal{F}_{0}\mathcal{(}x)\left \langle 2\varphi (x,t)-%
\frac{\varphi ^{2}(x,t)}{2}\right \rangle \dot{\psi}(x,t)\right \} dxdt=0
\label{33g}
\end{eqnarray}%
and 
\begin{eqnarray}
&&\int_{0}^{T}\left \{ \int_{\Omega }\rho (x)\mathbf{\ddot{u}}(x,t)\cdot 
\mathbf{\dot{w}}(x,t)+\left[ 1-\breve{\varphi}(x,t)\right] ^{2}\mathbf{A}%
(x)\nabla \mathbf{u}(x,t)\right.   \notag \\
&&\qquad \left. \cdot \nabla \mathbf{\dot{w}}(x,t)\right \}
dxdt=\int_{0}^{T}\int_{\Omega }\rho (x)\mathbf{b}(x,t)\cdot \mathbf{\dot{w}}%
(x,t)dxdt.  \label{33h}
\end{eqnarray}
\end{definition}

In order to obtain the well position of the differential problem (\ref{33b})-%
\ref{33f}) it is crucial to prove the following maximum theorem.

\begin{theorem}
Any solution of the equation (\ref{33b}), with initial conditions $0\leq
\varphi _{0}(x)\leq 1$, with boundary condition $\left. \nabla \varphi
(x,t)\right \vert _{\partial \Omega }=0,$ satisfies the restriction%
\begin{equation}
0\leq \varphi (x,t)\leq 1,~for\ all\ ~t>0~and\ a.e.\text{ }x\in \Omega .
\label{33i}
\end{equation}

\begin{proof}
Multiplying the equation (\ref{33b}) by $\varphi _{\_}(x,t)=\max \left[
-\varphi (x,t),0\right] $ and integrating on $\Omega ,$ we obtain%
\begin{eqnarray}
&&\int_{\Omega }\left( \rho \varphi _{\_}\varphi _{t}+\frac{1}{\kappa }%
\nabla \varphi _{\_}\cdot \nabla \varphi +\varphi _{\_}\mathcal{F}%
_{0}\left \langle 2\varphi -\frac{\varphi ^{2}}{2}\right \rangle \right) dx 
\notag \\
&&\qquad =-\int_{\Omega }\left \langle -1\right \rangle \varphi
_{\_}\int_{0}^{t}\left[ 1-\breve{\varphi}(\tau )\right] \mathbf{A}\nabla 
\mathbf{u}(\tau )\cdot \nabla \mathbf{\dot{u}}(\tau )d\tau dx.  \label{33l}
\end{eqnarray}%
It is easy to prove from the initial condition $\varphi _{0}(x)\in \lbrack
0,1]$, that%
\begin{equation*}
\int_{\Omega }\left \langle -1\right \rangle \varphi _{\_}(t)\int_{0}^{t}\left[
1-\breve{\varphi}(\tau )\right] \mathbf{A}\nabla \mathbf{u}(\tau )\cdot
\nabla \mathbf{\dot{u}}(\tau )d\tau dx=0;
\end{equation*}%
then, because $\varphi _{\_}\varphi _{t}=-\varphi _{\_}\varphi _{\_t}$, we
obtain from (\ref{33l}) for any $T>0$%
\begin{eqnarray*}
&&\int_{0}^{T}\int_{\Omega }\left[ \rho \varphi _{\_}(t)\varphi _{\_t}(t)+%
\frac{1}{\kappa }\nabla \varphi _{\_}(t)\cdot \nabla \varphi _{\_}(t)\right. 
\\
&&\qquad \left. -\varphi _{\_}(t)\mathcal{F}_{0}\left \langle 2\varphi (t)-%
\frac{\varphi ^{2}}{2}(t)\right \rangle \right] dxdt=0,
\end{eqnarray*}%
from which%
\begin{equation}
\int_{\Omega }\left \{ \frac{1}{2}\rho \left[ \varphi _{\_}(T)\right] ^{2}+%
\frac{1}{\kappa }\int_{0}^{T}\left[ \nabla \varphi _{\_}(t)\right]
^{2}dt\right \} dx=0,  \label{33q}
\end{equation}%
because 
\begin{equation*}
\int_{0}^{T}\int_{\Omega }\varphi _{\_}(t)\mathcal{F}_{0}\left \langle
2\varphi (t)-\frac{\varphi ^{2}}{2}(t)\right \rangle dxdt=0.
\end{equation*}%
Thus, from (\ref{33q}), $\varphi _{\_}(x,T)=0$ for any $T\in \lbrack
0,\infty )$ and $\nabla \varphi _{\_}(t)=\mathbf{0}$ for any $t\in \left[ 0,T%
\right] $\ a.e. $x\in \Omega .$

Hence, the restriction that $\varphi (x,t)\geq 0~$ is satisfied in (\ref{33b}%
).

To show that $\varphi (x,t)\leq 1$, let us to consider the function 
\begin{equation*}
(\varphi -1)_{+}=\left \{ 
\begin{array}{ll}
\varphi -1 & \forall \varphi >1 \\ 
0 & \forall \varphi \leq 1%
\end{array}%
\right.
\end{equation*}%
and suppose that $\varphi (x,t)>1$.

Thus, multiplying (\ref{33b}) by $(\varphi -1)_{+}$, after an integration on 
$\Omega ,$ we have%
\begin{eqnarray}
&&\int_{\Omega }\left[ \rho (\varphi -1)_{+}\, \varphi _{t}+\frac{1}{\kappa }%
\nabla (\varphi -1)_{+}\cdot \nabla \varphi +(\varphi -1)_{+}\mathcal{F}%
_{0}\left \langle 2\varphi -\frac{\varphi ^{2}}{2}\right \rangle \right] dx 
\notag \\
&&\; \qquad =-\left \langle -1\right \rangle \int_{\Omega }(\varphi
-1)_{+}\int_{0}^{t}\left[ 1-\breve{\varphi}(\tau )\right] \mathbf{A}\nabla 
\mathbf{u}(\tau )\cdot \nabla \mathbf{\dot{u}}(\tau )d\tau dx  \label{33qq}
\end{eqnarray}%
and hence, by integrating on $\left[ 0,T\right] $, it follows that%
\begin{eqnarray*}
&&\, \int_{0}^{T}\int_{\Omega }\left[ \rho (\varphi -1)_{+}(\varphi -1)_{t}+%
\frac{1}{\kappa }\nabla (\varphi -1)_{+}\cdot \nabla (\varphi -1)\right.  \\
&&\, \qquad \qquad \qquad \qquad \qquad \left. +(\varphi -1)_{+}\mathcal{F}%
_{0}\left \langle 2\varphi -\frac{\varphi ^{2}}{2}\right \rangle \right] dxdt
\\
&&\quad =-\int_{0}^{T}\int_{\Omega }\left \langle -1\right \rangle (\varphi
-1)_{+}\int_{0}^{t}\left[ 1-\breve{\varphi}(\tau )\right] \mathbf{A}\nabla 
\mathbf{u}(\tau )\cdot \nabla \mathbf{\dot{u}}(\tau )d\tau dxdt=0,\qquad 
\end{eqnarray*}%
since by hypothesis $\varphi (x,t)>1$ , then $\left \langle -1\right \rangle
(\varphi -1)_{+}=0$.

By integrating the first term and observing that in the third term also $%
G^{\prime }(\varphi )=\left \langle 2\varphi -\frac{\varphi ^{2}}{2}%
\right \rangle =0$ under the hypothesis $\varphi (x,t)>1$, this equation
reduces to%
\begin{equation}
\int_{\Omega }\left \{ \frac{1}{2}\rho \left[ (\varphi -1)_{+}(T)\right] ^{2}+%
\frac{1}{\kappa }\int_{0}^{T}\left[ \nabla (\varphi -1)_{+}(t)\right]
^{2}dt\right \} dx=0,  \label{33qr}
\end{equation}%
that is the sum of two positive quantities; therefore, each of these must be
equal to zero. Thus, it follows that $(\varphi -1)_{+}(x,T)=0$ for any $T\in
\lbrack 0,\infty )$ and also $\nabla (\varphi -1)_{+}(x,t)=\mathbf{0}$\ for
any $t\in \left[ 0,T\right] .\ $Hence, the hypothesis $\varphi (x,t)>1$ is
not satisfied; on the contary the inequality $\varphi (x,t)\leq 1$ holds\
and the restriction (\ref{33i}) has been proved.
\end{proof}
\end{theorem}

\bigskip

The thesis of this theorem can be confirmed by the simulations represented
in Fig. 3, and 4, corresponding to several selections of parameters

\begin{center}
\includegraphics[scale=0.55]{{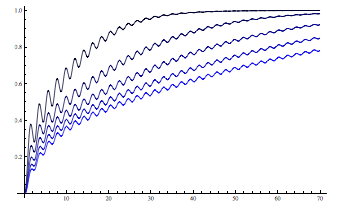}} \ %
\includegraphics[scale=0.55]{{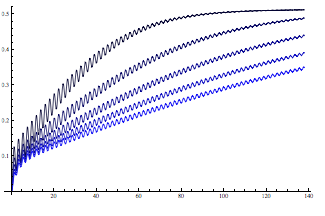}}

Fig.3. The first graphic describes the behavior of the phase field $\varphi $
corresponding to different values of the density $\bar{\rho}_{i}=1,2,3,4,5$
and a time $t=140$, with the same frequency $\bar{\omega}$. In the second
graphic, the frequency is the same $\bar{\omega}$, but the densities $\rho
_{i}=2\bar{\rho}_{i}.$

\includegraphics[scale=0.7]{{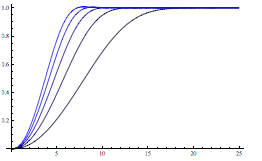}} \ %
\includegraphics[scale=0.72]{{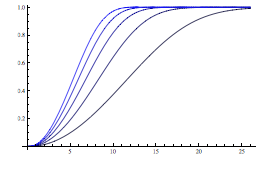}}

Fig. 4. These graphics describe the behavior of $\varphi $ for different
values of the frequency $\bar{\omega}_{i}=1,2,3,4,5$, with the same time $%
t=26$, but with two values of the density.

\bigskip
\end{center}

Let us now study a uniqueness theorem for the null solution. So that, we
consider zero initial data and supplies.

\begin{theorem}
The system (\ref{33g})-\ref{33h}) with zero initial data and supplies admits
only the null solution.
\end{theorem}

\begin{proof}
From the system (\ref{33g})-\ref{33h}) with $\psi =\varphi $ and $\mathbf{%
w=u,}$ and initial and boundary conditions%
\begin{equation}
\mathbf{u}(x,0)=\mathbf{0},~\mathbf{\dot{u}}(x,0)=\mathbf{0},~\varphi
(x,0)=0,  \label{33o}
\end{equation}%
\begin{equation}
\left. \mathbf{u}(x,t)\right \vert _{\partial \Omega }=\mathbf{0},~\  \left.
\nabla \varphi (x,t)\right \vert _{\partial \Omega }=0,  \label{33p}
\end{equation}%
we obtain 
\begin{eqnarray}
&&\int_{0}^{T}\int_{\Omega }\left \{ \rho (x)\left[ \dot{\varphi}(x,t)\right]
^{2}+\frac{1}{2}\frac{1}{\kappa (x)}\frac{\partial }{\partial t}\left[
\nabla \varphi (x,t)\right] ^{2}\right.   \notag \\
&&\quad +\left \langle -1\right \rangle \int_{t_{0}}^{t}\left[ 1-\breve{\varphi%
}(x,\tau )\right] \mathbf{A}(x)\nabla \mathbf{u}(x,\tau )\cdot \nabla 
\mathbf{\dot{u}}(x,\tau )d\tau \dot{\varphi}(x,t)  \notag \\
&&\quad \left. +\mathcal{F}_{0}\mathcal{(}x)\left \langle 2\varphi (x,t)-%
\frac{\varphi ^{2}(x,t)}{2}\right \rangle \right \} \dot{\varphi}(x,t)dxdt=0
\label{33m}
\end{eqnarray}%
and%
\begin{equation}
\int_{0}^{T}\int_{\Omega }\rho (x)\mathbf{\ddot{u}}(x,t)\cdot \mathbf{\dot{u}%
}(x,t)+\left[ 1-\breve{\varphi}(x,t)\right] ^{2}\mathbf{A}(x)\nabla \mathbf{u%
}(x,t)\cdot \nabla \mathbf{\dot{u}}(x,t)dxdt=0.  \label{33n}
\end{equation}%
Summing these equations and taking into account the previous result $0\leq
\varphi (x,t)\leq 1$ and the corresponding expressions of the interested
quantities, we obtain%
\begin{eqnarray}
&&\int_{0}^{T}\int_{\Omega }\left( \frac{1}{2}\frac{\partial }{\partial t}%
\left \{ \rho (x)\mathbf{\dot{u}}^{2}(x,t)+\left[ 1-\varphi (x,t)\right] ^{2}%
\mathbf{A}(x)\nabla \mathbf{u}(x,t)\cdot \nabla \mathbf{u}(x,t)\right.
\right.   \notag \\
&&\; \left. \left. +\frac{1}{\kappa (x)}\left[ \nabla \varphi (x,t)\right]
^{2}\right \} +\mathcal{F}_{0}\mathcal{(}x)\left( 2\varphi (x,t)-\frac{%
\varphi ^{2}(x,t)}{2}\right) \dot{\varphi}(x,t)\right) dxdt  \notag \\
&=&\int_{0}^{T}\int_{\Omega }\left \{ -\rho (x)\left[ \dot{\varphi}(x,t)%
\right] ^{2}+\frac{1}{2}\dot{\varphi}(x,t)\int_{0}^{t}\dot{\varphi}(x,\tau )%
\mathbf{A}(x)\nabla \mathbf{u}(x,\tau )\right.   \notag \\
&&\; \left. \cdot \nabla \mathbf{u}(x,\tau )d\tau -\frac{1}{2}\dot{\varphi}%
(x,t)\left[ 1-\varphi (x,t)\right] \mathbf{A}(x)\nabla \mathbf{u}(x,t)\cdot
\nabla \mathbf{u}(x,t)\right \} dxdt.\qquad \quad   \label{33r}
\end{eqnarray}%
Then, using the initial conditions (\ref{33o}), there exists a $%
T_{\varepsilon }>0$ such that for all $t\in \left[ 0,T_{\varepsilon }\right] 
$%
\begin{eqnarray}
0 &\geq &\dot{\varphi}(x,t)\left \{ \int_{0}^{t}\dot{\varphi}(x,\tau )\mathbf{%
A}(x)\nabla \mathbf{u}(x,\tau )\cdot \nabla \mathbf{u}(x,\tau )d\tau \right. 
\label{33ww} \\
&&-\left. \left[ 1-\varphi (x,t)\right] \mathbf{A}(x)\nabla \mathbf{u}%
(x,t)\cdot \nabla \mathbf{u}(x,t)\right \} ;  \notag
\end{eqnarray}%
then, we have by (\ref{33r}) and by the initial and boundary conditions (\ref%
{33o})-(\ref{33p}), that%
\begin{equation}
\mathbf{u}(x,t)=\mathbf{0},~\varphi (x,t)=0,~\ t\in \left[ 0,T_{\varepsilon }%
\right] ~and~a.e.\text{ }x\in \Omega .  \label{33w}
\end{equation}%
Finally, after an iteration procedure, we can prove that (\ref{33w}) is
satisfied for any $t\in \left[ 0,\infty \right] .$
\end{proof}

\section{Thermo-mechanical fatigue life}

It is apparent that fatigue may be affected by the temperature (see \cite{KC}%
, \cite{MR}). Thus, fatigue life of a material could change under different
temperatures, especially, when the variations of temperature are frequent
and with high gradients or thermal shock.

In this framework, we have to introduce the heat equation. So that, we
consider the First Law of Thermodynamics%
\begin{equation}
\rho (x)~\dot{e}(x,t)=\mathcal{P}_{m}^{i}(x,t)+\mathcal{P}_{h}^{i}(x,t)+%
\mathcal{P}_{s}^{i}(x,t),  \label{34}
\end{equation}%
where $e$ is the internal energy, $\mathcal{P}_{m}^{i}$ the internal
mechanical power, $\mathcal{P}_{h}^{i}$ the internal heat power and $%
\mathcal{P}_{s}^{i}$ the internal structural power. Moreover, in this paper
we confine the study to a \textit{simple material }(see \cite{TN}, \cite{FLN}%
), for which the internal mechanical work $\mathcal{P}_{m}^{i}$ is given by%
\begin{equation}
\mathcal{P}_{m}^{i}=\mathbf{T}\cdot \mathbf{L}  \label{34a}
\end{equation}%
With a view to well define the differential system, we introduce the motion
equation%
\begin{equation}
\rho (x)\frac{\partial \mathbf{v}(x,t)}{\partial t}=\nabla \cdot \left[ 
\mathbf{T}(x,t)+\varkappa \theta \mathbf{I}\right] +\rho (x)\mathbf{b}%
(x,t),~\  \varkappa >0,  \label{35}
\end{equation}%
where $\theta $~is the absolute temperature and $\mathbf{I}$\ is the
identity tensor, the heat equation%
\begin{equation}
\mathcal{P}_{h}^{i}(x,t)=-\nabla \cdot \mathbf{q(}x,t)+\rho (x)r(x,t)
\label{36}
\end{equation}%
and the Ginzburg-Landau equation%
\begin{eqnarray}
\rho (x)\frac{\partial \varphi (x,t)}{\partial t} &=&\nabla \cdot \frac{1}{%
\kappa (x)}\nabla \varphi (x,t)  \notag \\
&&-\mathcal{F(}x,t)F^{\prime }(\varphi (x,t))-\mathcal{F}_{0}\mathcal{(}%
x)G^{\prime }(\varphi (x,t)).  \label{37a}
\end{eqnarray}%
The fatigue $\mathcal{F(}\hat{\sigma}(t))$ is now defined by%
\begin{eqnarray}
\mathcal{F(}\hat{\sigma}(t)) &=&\int_{t_{0}}^{t}\left[ 1-\breve{\varphi}%
(\tau )\right] \left[ -\frac{\mathcal{P}_{h}^{i}(\tau )}{\theta (\tau )}+%
\mathbf{q}\cdot \nabla \theta ^{-1}(\tau )\right] d\tau   \notag \\
&=&\int_{t_{0}}^{t}\left[ 1-\breve{\varphi}(\tau )\right] \left \{ [\mathbf{%
\tilde{T}}(\sigma (\tau ),P(\tau ))\cdot \nabla \mathbf{v}(\tau )+\mathcal{P}%
_{s}^{i}(\tau )\right.   \notag \\
&&\quad \left. -\rho (x)\dot{e}(\tau )]\theta ^{-1}(\tau )+\mathbf{q}(\tau
)\cdot \nabla \theta ^{-1}(\tau )\right \} d\tau ,  \label{38}
\end{eqnarray}%
which on closed cycles assumes the form%
\begin{eqnarray}
&&\mathcal{F(}\hat{\sigma}(t))=\oint_{t_{0}}^{t}\left[ 1-\breve{\varphi}%
(\tau )\right] [\mathbf{\tilde{T}}(\sigma (\tau ),P(\tau ))\cdot \mathbf{L}%
(\tau )\theta ^{-1}(\tau )  \notag \\
&&\,+\mathcal{P}_{s}^{i}(\tau )\theta ^{-1}(\tau )+\mathbf{q}(\tau )\cdot
\nabla \theta ^{-1}(\tau )]d\tau -\oint_{t_{0}}^{t}\rho (x)\left[ 1-\breve{%
\varphi}(\tau )\right] \frac{\dot{e}(\tau )}{\theta (\tau )}d\tau .\quad 
\label{39}
\end{eqnarray}

Then, by (\ref{37a}) the balance of the structural power provides%
\begin{eqnarray*}
&&\, \nabla \cdot \left[ \frac{1}{\kappa }(x)\dot{\varphi}(x,t)\nabla
\varphi (x,t)\right] =\rho (x)\dot{\varphi}^{2}(x,t)+\frac{1}{\kappa (x)}%
\nabla \varphi (x,t)\cdot \nabla \dot{\varphi}(x,t) \\
&&\qquad \qquad \qquad \qquad \qquad +\mathcal{F(}\hat{\sigma}(t))\dot{F}%
(\varphi (x,t))+\mathcal{F}_{0}\mathcal{(}x)\dot{G}(\varphi (x,t)).
\end{eqnarray*}%
So, the internal structural power is defined by%
\begin{eqnarray}
\mathcal{P}_{s}^{i}(x,t) &=&\rho (x)\dot{\varphi}^{2}(x,t)+\frac{1}{\kappa
(x)}\nabla \varphi (x,t)\cdot \nabla \dot{\varphi}(x,t)  \notag \\
&&+\mathcal{F(}\hat{\sigma}(t))\dot{F}(\varphi (x,t))+\mathcal{F}_{0}%
\mathcal{(}x)\dot{G}(\varphi (x,t)).  \label{40c}
\end{eqnarray}

Finally, by (\ref{34}), (\ref{34a}) and (\ref{36}) we obtain the heat
equation%
\begin{equation*}
\rho \dot{e}-\mathbf{T\cdot L-}\mathcal{P}_{s}^{i}=-\nabla \cdot \mathbf{q}%
+\rho r,
\end{equation*}%
where $\mathcal{P}_{s}^{i}(x,t)$\ is given by (\ref{40c}).

This equation for an elastic material takes the wording 
\begin{equation*}
\rho c\dot{\theta}-\rho (x)\dot{\varphi}^{2}(x,t)-\mathcal{F(}\hat{\sigma}%
(t))\dot{F}(\varphi (x,t))=-\nabla \cdot \mathbf{q}+\rho r
\end{equation*}%
and from (\ref{40c}), $\mathcal{F(}\hat{\sigma}(t))$ assumes the form%
\begin{eqnarray}
\mathcal{F(}\hat{\sigma}(t)) &=&\int_{t_{0}}^{t}(\left[ 1-\breve{\varphi}%
(\tau )\right] \mathbf{A}\nabla \mathbf{u}(\tau )\cdot \nabla \mathbf{\dot{u}%
}(\tau )\theta ^{-1}(\tau )  \notag \\
&&+\mathbf{q}(\tau )\cdot \nabla \theta ^{-1}(\tau )(d\tau .  \label{41}
\end{eqnarray}

\begin{acknowledgement}
The author has been partially supported by Italian G.N.F.M. - I.N.D.A.M and
by University of Bologna. The Figures 3 and 4 were generated using the
software package Mathematica (ver. 8.0.4). Finally, the author is grateful
to professors Caputo and Fremond for their suggestions and advises.
\end{acknowledgement}

\end{document}